\documentclass[a4paper]{article}

\usepackage{amssymb}
\usepackage[latin1] {inputenc}
\usepackage{amsmath}
\usepackage{amsthm}
\usepackage{amsfonts}
\usepackage{amssymb}
\usepackage{lscape}
\usepackage{graphicx}
\usepackage{verbatim}
\usepackage{url}
\usepackage{hyperref}
\usepackage{anysize}
\marginsize{3.2cm}{3.2cm}{2.7cm}{2.7cm}

\newcommand{\I}{\hspace{0.45cm}}
\newcommand{\II}{\I\I}

\newtheorem{thm}{Theorem}
\newtheorem{lem}[thm]{Lemma}
\newtheorem{cor}[thm]{Corollary}

\newtheorem{fact}[thm]{Fact}

\begin{document}

\title{Complexity dichotomy on partial grid recognition\thanks{An extended abstract of this paper was presented at ISCO 2010, the International Symposium on Combinatorial Optimization, to appear in Electronic Notes in Discrete Mathematics.}}

\author{
Vin\'icius G.~P.~de S\'a\thanks{%
DCC/IM, Univ.~Federal do Rio de Janeiro (UFRJ), Brazil.
Research supported by FAPERJ grant E-26/102.083/2009.}
\and
Guilherme D. da Fonseca\thanks{%
Univ.~Federal do Estado do Rio de Janeiro (UNIRIO), Brazil.
Research supported by FAPERJ grant E-26/110.091/2010.}
\and
Raphael Machado\thanks{%
Instituto Nacional de Metrologia, Normaliza\c{c}\~ao e Qualidade Industrial (INMETRO), Brazil.}
\and
Celina M. H. de Figueiredo\thanks{%
COPPE, Univ.~Federal do Rio de Janeiro (UFRJ), Brazil.
Research supported by FAPERJ grant E-26/102.706/2008 and CNPq grant 472144/2009-0.}
}

\maketitle

\begin{abstract}
Deciding whether a graph can be embedded in a grid using only unit-length edges is NP-complete, even when restricted to binary trees. However, it is not difficult to devise a number of graph classes for which the problem is polynomial, even trivial. A natural step, outstanding thus far, was to provide a broad classification of graphs that make for polynomial or NP-complete instances. We provide such a classification based on the set of allowed vertex degrees in the input graphs, yielding a full dichotomy on the complexity of the problem. As byproducts, the previous NP-completeness result for binary trees was strengthened to strictly binary trees, and the three-dimensional version of the problem was for the first time proven to be NP-complete. Our results were made possible by introducing the concepts of consistent orientations and robust gadgets, and by showing how the former allows NP-completeness proofs by local replacement even in the absence of the latter.
\end{abstract}

\section{Introduction}

A \emph{grid} $G_{M\times N}$ has vertex set $V(G_{M\times N})=\{(i,j): 1 \leq i \leq M, 1 \leq j \leq N\}$, and edge set $E(G_{M\times N})=\{(i,j)(k,l): |i-k|+|j-l|=1, (i,j),(k,l)\in V(G_{M\times N})\}$. Grids are often thought of in terms of their usual graphical representation, where vertices are the intersection points of lines that cross over each other in a regular pattern, as illustrated in Figure~\ref{figGridGraphs}(a). Grids are planar bipartite graphs.

A \emph{unit-length embedding} (or \emph{embedding}, for short, throughout the whole text) is a mapping from the vertex set of a graph $G$ to a subset of the points of a grid, along with an incidence-preserving assignment of the edges of $G$ to unit-length grid segments. We refer to such set of points and unit-length segments as a grid \emph{drawing}. Two embeddings are equal if they correspond to the same drawing, short of rotation, translation and reflection.

A \emph{partial grid} is any subgraph (not necessarily induced) of a grid, and can also be characterized as a graph that admits a unit-length embedding. Grid embeddings are widely studied due to applications in VLSI design~\cite{Ull} and simulation of parallel architectures~\cite{LMPPS03}. Unfortunately, deciding whether a graph admits a unit-length embedding is {NP-com\-plete}~\cite{IPL87}, even when restricted to binary trees~\cite{IPL89}. Indeed the so-called \emph{logic engine paradigm} for proving the NP-hardness of problems in Graph Drawing is described in~\cite{BETT99}, where the seminal references~\cite{IPL87,IPL89} and further applications~\cite{EW96,EW96b} are discussed. On the other hand, in the context of Graph Theory, the recognition of partial grid graphs is often stated as an open problem~\cite{graphclass,CM05}.

Let $G$ be a graph. The vertex and edge sets of $G$ are denoted $V(G)$ and $E(G)$, respectively, and $d_G(v)$ stands for the degree of vertex $v$ in $G$. Now let $D$ be a set of integers. We say $G$ is a $D$-\emph{graph} if, for all $v \in V(G)$, we have $d_G(v) \in D$, e.g.~paths are \{1,2\}-graphs, cycles are \{2\}-graphs, a complete graph on $n$ vertices is a \{$n-1$\}-graph etc. Figure~\ref{figGridGraphs}(b) illustrates a \{1,2,4\}-tree.

The \textsc{Partial-Grid Recognition} problem (\textsc{PGR}) asks whether a graph $G$ is a partial grid. In this paper, we establish the problem's complexity dichotomy into polynomial and NP-complete when the input is restricted to $D$-graphs, for every $D \subseteq \{1,2,3,4\}$, thus exhausting all possible sets whose elements can be found as vertex degrees in partial grid graphs. All graphs we consider are connected, since the problem can be solved independently for each connected component of a disconnected graph. Moreover, we will certainly use the facts that, (i) if the problem is NP-complete for $D$-trees, then it is also NP-complete for $D'$-trees, $D' \supset D$, and for $D$-graphs and $D'$-graphs---allowing cycles---as well (\emph{superset property}); and, analogously, if the problem is polynomial for $D$-graphs, then it is also polynomial for $D'$-graphs, $D' \subset D$, and for $D$-trees and $D'$-trees as well (\emph{subset property}).

In Section~\ref{secPrelim}, we revisit the seminal NP-completeness proofs and define the basic concepts for the sections to come.

Section~\ref{secCD} is the core of the present paper, addressing the complexity of each outstanding case---we either prove its NP-completeness, or state its triviality, or give a polynomial-time algorithm when applicable.

Additionally, motivated by recent advances in three-dimensional chip manufacturing~\cite{DANG08,SAKUMA08,SAKUMA09}, we consider the natural three-dimensional version of the problem in Section~\ref{sec3d}. We then illustrate the power of our techniques by proving simple theorems that settle the complexity classes of recognizing 3d partial grids for the vast majority of acceptable input degrees.

Section~\ref{secConclusion} closes the paper with concluding remarks and open problems.

\begin{figure}[!t]
\begin{center}

\setlength{\unitlength}{1mm}

\includegraphics[scale=0.69]{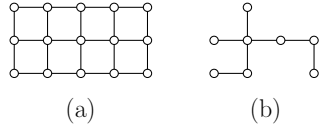}

\caption{\label{figGridGraphs}(a) The grid $G_{3,5}$. (b) Unit-length embedding for a \{1,2,4\}-tree.}

\end{center}
\end{figure}

\section{Consistent orientations and immersibility}\label{secPrelim}

Let $G$ be a graph. We say $f_G: E(G) \rightarrow \{0,1\}$ is a \emph{consistent orientation} for $G$ when it holds that, if $G$ is a partial grid, then there is an embedding for $G$ where every edge in $\{xy \in E(G): f_G(xy)=0\}$ is drawn horizontally, and every edge in $\{xy \in E(G): f_G(xy)=1\}$ is drawn vertically on the grid. Note that, if $G$ is not a partial grid, then any boolean function is a consistent orientation\linebreak for $G$.

We say two graphs $G_1,G_2$ have the same \emph{immersibility} if (i) both $G_1$ and $G_2$ are partial grids, or (ii) neither $G_1$ or $G_2$ is a partial grid.

In~\cite{IPL87}, Bhatt and Cosmadakis proved that deciding the existence of unit-length embeddings for arbitrary trees is NP-complete. Their proof was based on the reduction of the well-known NP-complete problem \textsc{Not-All-Equal} 3CNF SAT (not-all-equal conjunctive-normal-form satisfiability with 3 literals per clause) to the problem of deciding the existence of a unit-length embedding for a special \{1,2,4\}-tree they define, called the \emph{extended skeleton} (see Figure~\ref{figSkeleton}). This problem is referred to as the \textsc{Bhatt-Cosmadakis} problem.

\begin{figure}[!t]
\begin{center}

\setlength{\unitlength}{1mm}

\includegraphics[scale=0.69]{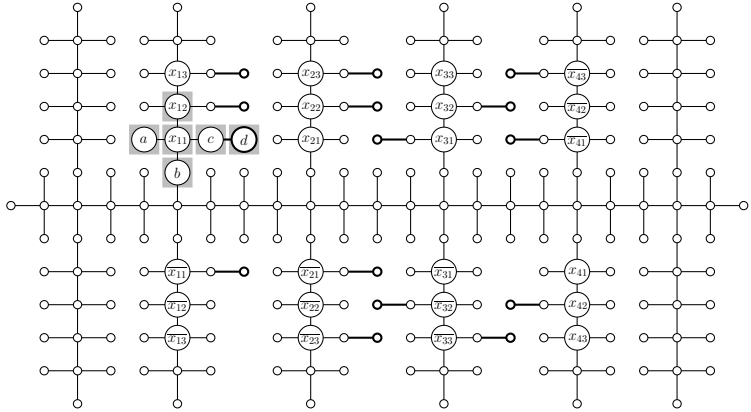}

\caption{\label{figSkeleton}Grid embedding for Bhatt and Cosmadakis's extended skeleton $S_\varphi$ associated to the 3CNF formula $\varphi = (\overline{x_2} \lor x_3 \lor \overline{x_4}) \land (x_1 \lor x_2 \lor x_4) \land (x_1 \lor \overline{x_3} \lor \overline{x_4})$. The existence of such embedding for $S_\varphi$ relates to the existence of a satisfying assignment for $\varphi$, namely $(x_1,x_2,x_3,x_4) = (T,T,T,F)$.}

\end{center}
\end{figure}

Though we will not give the details of such special tree here, the following fact is of utmost importance:

\begin{fact}\label{pro:known_orientations}
If $S_\varphi$ is an extended skeleton, then a consistent orientation for $S_\varphi$ can be determined in polynomial time.
\end{fact}
\begin{proof}
An extended skeleton $S_\varphi$ comprises a subgraph $Z_\varphi$, called \emph{skeleton}, and a set of edges in $S_\varphi \setminus Z_\varphi$, called \emph{flags} (flags are shown in bold lines, in Figure~\ref{figSkeleton}). The skeleton is itself a partial grid which cannot accept two distinct embeddings, due to the rigidity granted by its main and transversal \emph{spinal cords} (the main spinal cord can be easily pinpointed in Figure~\ref{figSkeleton}---it comprises the long path of 4-degree vertices drawn in a straight horizontal line). The flags, on their turn, can only be embedded with the same orientation as the edges in the main spinal cord. On these grounds, the algorithm in Figure~\ref{algSphi} gives a consistent orientation for $S_\varphi$ in polynomial time, and Fact~\ref{pro:known_orientations} follows.
\end{proof}

\begin{figure}[!tb]\begin{flushleft}
\begin{sffamily}\medskip\setlength{\parindent}{0pt}\setlength{\parskip}{0pt}

\textbf{Extended\_Skeleton\_Consistent\_Orientation}
$(S_\varphi:~\textrm{extended skeleton})$

~

\begin{tabular}{ll}

\textbf{1.} & \textbf{for each} $uv$ \textbf{in} $E(S_\varphi)$ \textbf{do} \\

\textbf{1.1.} & \I $f_{S_\varphi}(uv) \gets -1$~~\emph{// mark the orientation of all edges as undefined}\\

\textbf{2.} & $F \gets \{uv \in E(S_\varphi): d_{S_\varphi}(u)=1,d_{S_\varphi}(v)=2\} $~~\emph{// flags}\\

\textbf{3.} & $Z_\varphi \gets S_\varphi \setminus F$~~\emph{// skeleton}\\

\textbf{4.} & $P \gets \{p: p~\textrm{is a maximal path of 4-degree vertices in}~Z_\varphi\}$\\

\textbf{5.} & \textbf{let}~$m$~be the only path in $P$ containing some vertex with two\\
            & ~~2-degree neighbors in $Z_\varphi$~~\emph{// main spinal cord}\\

\textbf{6.} & $T \gets P \setminus \{m\}$~~\emph{// transversal spinal cords}\\

\textbf{7.} & \textbf{for each} $uv$ \textbf{in} $m$ \textbf{do} \\

\textbf{7.1.} & \I $f_{S_\varphi}(uv) \gets 0$~~\emph{// main spinal cord oriented horizontally}\\

\textbf{8.} & \textbf{for each} $t$ \textbf{in} $T$ \textbf{do} \\

\textbf{8.1.} & \I \textbf{for each} $uv$ \textbf{in} $t$ \textbf{do} \\

\textbf{8.1.1} & \II $f_{S_\varphi}(uv) \gets 1$~~\emph{// transversal spinal cords oriented vertically}\\

\textbf{9.} & \textbf{for each} $uv$ \textbf{in} $Z_\varphi$ s.t.~$d_{Z_\varphi(u)}=2$ \textbf{do} \\

\textbf{9.1.} & \I $f_{S_\varphi}(uv) \gets 1$~~\emph{// edges connecting transversal to main spinal cords}\\
              & \I \phantom{$f_{S_\varphi}(uv) \gets 1$~~}\emph{// oriented vertically} \\

\textbf{10.} & \textbf{for each} $uv$ \textbf{in} $E(Z_\varphi)$ s.t.~$d_{Z_\varphi}(u)=4, f_{S_\varphi}(uv)=-1$ \textbf{do} \\

\textbf{10.1.} & \I \textbf{if} there exist $w,z \in V(Z_\varphi)$ s.t.~$f_{S_\varphi}(uw)=0, f_{S_\varphi}(uz)=0$ \textbf{then} \\

\textbf{10.1.1} & \II $f_{S_\varphi}(uv) \gets 1$~~\emph{// at most two horizontal edges allowed!}\\

\textbf{10.2.} & \I \textbf{else}\\

\textbf{10.2.1} & \II $f_{S_\varphi}(uv) \gets 0$\\

\textbf{11.} & \textbf{for each} $uv$ \textbf{in} $F$ \textbf{do} \\

\textbf{11.1.} & \I $f_{S_\varphi}(uv) \gets 0$~~\emph{// flags oriented horizontally}\\

\textbf{12.} & \textbf{return} $f_{S_\varphi}$

\end{tabular}

\end{sffamily}\end{flushleft}

\caption{\label{algSphi}Algorithm to determine consistent orientations for extended skeletons.}

\end{figure}

The seminal proof of Bhatt and Cosmadakis shows it is NP-complete to decide whether an extended skeleton is a partial grid, hence \textsc{PGR} is NP-complete for \{1,2,4\}-trees and, consequently, for \{1,2,3,4\}-trees.

The NP-completeness for \{1,2,3\}-trees (binary trees) was demonstrated by Gregori~\cite{IPL89}, who conceived an ingenious binary tree called the \emph{U-tree}. U-trees can be linked to one another by an edge between two of their vertices. Such vertices can be selected among four special vertices (in each U-tree), called the U-tree's \emph{interconnectors}. The U-tree is illustrated in Figure~\ref{figUtree}, where the \emph{horizontal interconnectors} $x,z$ and the \emph{vertical interconnectors} $y,w$ are indicated. Gregori proved that, by replacing each vertex of an extended skeleton $S_\varphi$ with a U-tree, the resulting \{1,2,3\}-tree $U(S_\varphi)$ is a partial grid if and only if the original \{1,2,3\}-tree $S_\varphi$ is. We call such operation the \emph{U-tree substitution}, and its output is the \emph{U-tree-transformed skeleton}. The U-tree substitution therefore \emph{preserves the immersibility} of extended skeletons, and the NP-completeness result followed. (We remark that Fact~\ref{pro:known_orientations} was used implicitly in Gregori's NP-completeness proof by local replacement, since a consistent orientation for the extended skeleton $S_\varphi$ is needed in~\cite{IPL89} to ensure that $U(S_\varphi)$ has the same immersibility as $S_\varphi$.)

\begin{figure}[!t]
\begin{center}

\setlength{\unitlength}{1mm}

\includegraphics[scale=0.69]{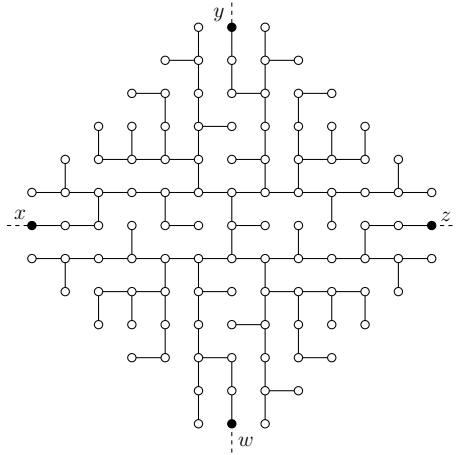}

\caption{\label{figUtree}Grid embedding for Gregori's U-tree.}

\end{center}
\end{figure}

Again, since the reader can find all the details of the U-tree substitution in the referenced paper, we underline the one single fact we will later depend upon:

\begin{fact}\label{pro:Utree_known_orientations}
If $U(S_\varphi)$ is a U-tree-transformed skeleton, then a consistent orientation for $U(S_\varphi)$ can be determined in polynomial time.
\end{fact}
\begin{proof}
Let $U$ be the U-tree graph and let $S_\varphi$ be the extended skeleton that ought to be submitted to a U-tree substitution. Unlike extended skeletons, whose elements conform with some associated boolean formula, $U$ is a fixed, predefined graph that accepts a small number of well-known embeddings (e.g.~the one given in Figure~\ref{figUtree}). Thus, a consistent orientation $f_U : E(U) \rightarrow \{0,1\}$ is known. Now, a U-tree-transformed skeleton $U(S_\varphi)$ is entirely made of interconnected U-trees, one for each vertex in the extended skeleton $S_\varphi$ being transformed. Thus, any edge $uv \in U(S_\varphi)$ is either an \emph{internal} edge and belongs to some copy of $U$ or is an \emph{external} edge linking two adjacent copies of $U$. It happens that, when $S_\varphi$ is submitted to a U-tree substitution, every edge of $S_\varphi$ that is horizontal, according to some (polynomially obtainable) consistent orientation $f_{S_\varphi} : E(S_\varphi) \rightarrow \{0,1\}$, yields an also horizontal alignment of the U-trees that replace its incident vertices. In other words, the vertices $u$ and $v$ between which an horizontal external edge exists in $U(S_\varphi)$ will have been selected among the horizontal interconnectors of the U-trees they belong to. Analogously, if the original edge in $S_\varphi$ has a vertical orientation according to $f_{S_\varphi}$, then the corresponding U-trees will be tied to one another via vertical interconnectors as well (and they will be linked to one another by a vertical external edge). This way, $90^\circ$-rotations of U-trees shall never take place, keeping the orientation of the internal edges untouched in $U(S_\varphi)$, exactly as given by $f_U$.

The function defined below combines both $f_U$ and $f_{S_\varphi}$ to obtain a consistent orientation $f_{U(S_\varphi)} : E(U(S_\varphi)) \rightarrow \{0,1\}$ for $U(S_\varphi)$, completing the proof.

\begin{displaymath}
f_{U(S_\varphi)}(uv) = \left\{ \begin{array}{ll}
   f_U(uv) & \textrm{if $uv$ is internal,}\\
   f_{S_\varphi}(s(u)s(v)) & \textrm{if $uv$ is external.}
     \end{array} \right.
\end{displaymath}

In the expression above, vertices $s(u),s(v) \in S_\varphi$ are those which were substituted by the U-trees that contain $u,v \in U(S_\varphi)$, respectively.

\end{proof}

\section{Complexity dichotomy}\label{secCD}

In the first part of this section, we prove that \textsc{PGR} is NP-complete for some input degree sets. The second part is devoted to the polynomially decidable cases. For the sake of clarity, in both parts we start the approach to each new case by stating the degree set under consideration thenceforth.

\subsection{NP-complete cases}

In the forthcoming proofs, we take for granted that \textsc{PGR} belongs to NP, regardless of the restrictions imposed to its input, as one can always check the soundness of a given embedding in polynomial time.

Let $G$ be a partial grid, $sv, vt \in E(G)$. If edges $sv$ and $vt$ appear as two consecutive segments of the same grid line (row or column) in some embedding of $G$, we say they constitute a pair of \emph{collinear} edges. Analogously, if there is an embedding of $G$ in which $sv$ and $vt$ appear with a $90^\textrm{o}$ angle between them, we say they form a pair of \emph{orthogonal} edges.

\subsection*{\textbf{\{2,3\}-graphs}}

In this section, we introduce a special \{2,3\}-graph called the \emph{double ladder}. Figure~\ref{figDoubleLadder}(a) presents its only existing embedding, where vertices $x,y,z,w$ are again seen as interconnectors, since edges connecting different double ladders can only be incident to two such vertices. We mark that the circular ordering of the interconnectors is fixed, that is, they cannot switch positions among themselves. For this reason, we say $x,z$ (and $w,y$ as well) constitute a pair of \emph{opposed} interconnectors, whereas all other pairs of interconnectors are \emph{consecutive}.

\begin{figure}[!t]
\begin{center}

\setlength{\unitlength}{1mm}

\includegraphics[scale=0.69]{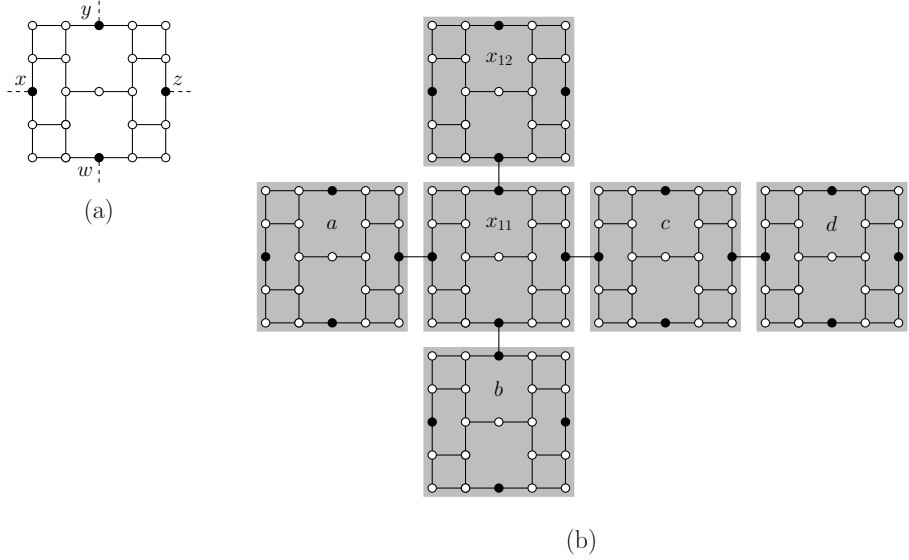}

\caption{\label{figDoubleLadder}(a) The \{2,3\} gadget (double ladder). (b) Double-ladder substitution.}

\end{center}
\end{figure}

Let $G$ be a graph. We define the \emph{double-ladder substitution} as the linear-time operation that obtains the graph $L(G)$ such that: (i) there is a bijection between each vertex $v$ in $G$ and a double ladder $l(v)$ in $L(G)$; and (ii) there is a bijection between each edge $uv$ in $G$ and an edge linking an interconnector of $l(u)$ to an interconnector of $l(v)$ in $L(G)$. Such interconnectors are said to have become \emph{active}. Figure~\ref{figDoubleLadder}(b) illustrates the result of a double-ladder substitution applied to the highlighted subgraph in Figure~\ref{figSkeleton}.

The double-ladder substitution does not necessarily preserve the immersibility of the original graph when the active interconnectors are chosen arbitrarily. The problem with structures like the double ladder, which present a fixed permutation of the interconnectors, is that they might not mimic the exact behavior of the original vertex they are meant to emulate. Indeed, if a pair of opposed (respectively, consecutive) interconnectors of $l(v)$ are chosen to link $l(v)$ to $l(s)$ and $l(t)$ during the double-ladder substitution, then the resulting graph $L(G)$ will only possibly admit embeddings in which those double ladders appear collinearly (resp.~orthogonally), thus destroying the equivalence between the immersibility of $G$ and that of $L(G)$ in case $sv,vt \in G$ happen not to be collinear (resp.~orthogonal) edges.

In order to preserve the immersibility of the original graph, it is mandatory that the choice of interconnectors match some feasible relative positioning of its edges, in case the graph is a partial grid. Although it may not be always easy to tell collinear from orthogonal pairs of edges in a given graph, Fact~\ref{pro:known_orientations} makes that a trivial task for extended skeletons.

\begin{lem}\label{lemDoubleLadder}Double-ladder substitution---with appropriately chosen {intercon\-nect\-ors}---preserves the immersibility of extended skeletons.
\end{lem}
\begin{proof}

Let $S_\varphi$ be an extended skeleton. By Fact~\ref{pro:known_orientations}, a consistent orientation $f_{S_\varphi}$ for all the edges of $S_\varphi$ can be determined in polynomial time. That is to say $f_{S_\varphi}$ provides us with a trustworthy relative positioning (collinear/orthogonal) of all edges incident to a common vertex. In order to match that positioning, it suffices that, for $sv,vt \in G$, the double-ladder substitution on graph $G$ employs a pair of opposed (resp.~consecutive) interconnectors of $l(v)$ to have it linked to $l(s)$ and $l(t)$ if $sv, vt$ are collinear (resp.~orthogonal).

Since a double ladder occupies a perfect $5\times5$ grid square in any unit-length embedding, the placement of the double ladder graphs in some embedding for $L(S_\varphi)$ shall always be met by a corresponding placement of $S_\varphi$'s vertices on a grid that is 5 times smaller. Edges linking one double ladder to another always occur between two adjacent $5\times5$ squares in the grid, therefore only edges of unit length will be required in the reduced grid. For the converse, we argue that, since the choice of interconnectors never disagrees with some consistent orientation of the edges of the extended skeleton, an embedding for an extended skeleton $S_\varphi$ will always lead to an embedding for $L(S_\varphi)$ in a grid that is 5 times larger.
\end{proof}

\begin{thm}
\textsc{PGR} is NP-complete for \{2,3\}- and \{2,3,4\}-graphs.
\end{thm}
\begin{proof}
Since \textsc{Bhatt-Cosmadakis} is NP-complete and it can be polynomially reduced---via double-ladder substitution on its input---to \textsc{PGR} restricted to \{2,3\}-graphs, the latter problem is NP-complete as well. The NP-completeness for \{2,3,4\}-graphs follows.
\end{proof}

The acyclic case does not apply, for there are no trees without leaves.

\subsection*{\textbf{\{2,4\}-graphs}}

To prove the NP-completeness of the problem for \{2,4\}-graphs, our strategy will be identical to that just seen for \{2,3\}-graphs. We introduce an appropriate substitution procedure that preserves the immersibility of the extended skeleton.

The replacement structure we use is a simple $C_4$, or \emph{square} (shown in Figure~\ref{figC4}(a), in solid lines), whose vertices are regarded as interconnectors. Surprisingly, the $C_4$ shall replace both vertices and edges of the original graph, in what we call the \emph{square substitution}. In the square substitution, each vertex $v$ of the original graph $G$ gives rise to a square $q(v)$ in the resulting graph $Q(G)$, and each edge $uv \in G$ corresponds to another $C_4$, call it $q(uv)$, in $Q(G)$, linking $q(v)$ to $q(u)$ using opposed interconnectors of $q(uv)$. Figure~\ref{figC4}(c) shows the result of the square substitution applied to the highlighted subgraph in Figure~\ref{figSkeleton}. Notice that it looks as though the original graph had been rotated $45^\circ$, as depicted in Figure~\ref{figC4}(b).

\begin{figure}[!t]
\begin{center}

\setlength{\unitlength}{1mm}

\includegraphics[scale=0.69]{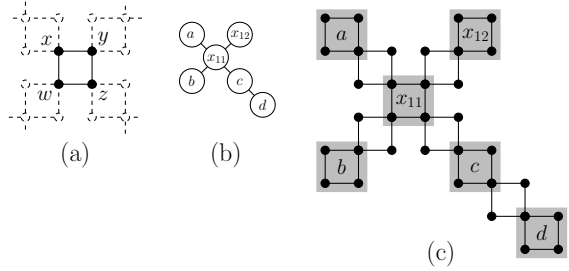}

\caption{\label{figC4}(a) The \{2,4\} gadget ($C_4$). (b) Rotation of $45^\circ$. (c) Square substitution.}

\end{center}
\end{figure}

\begin{lem}\label{lemSquare}Square substitution---with appropriately chosen interconnectors---{pre\-serves} the immersibility of extended skeletons.
\end{lem}
\begin{proof}
Here again, despite the fixed circular permutation of the interconnectors of a $C_4$, the foreknowledge of consistent orientations for extended skeletons (Fact~\ref{pro:known_orientations}) allows active interconnectors to be suitably chosen in $q(v)$. Let $S_\varphi$ be an extended skeleton and let $Q(S_\varphi)$ be the result of some such orientation-aware square substitution. We want to prove that $S_\varphi$ admits a unit-length embedding if and only if $Q(S_\varphi)$ does.

Suppose $S_\varphi$ is a partial grid graph. Then, there is a unit-length embedding $\Gamma$ for $S_\varphi$ such that the relative position of every pair of edges $sv,vt \in S_\varphi$ matches the (only) relative position of $q(sv),q(vt)$ allowed by that particular choice of interconnectors of $q(v)$. Now, it is always possible to obtain a unit-length embedding $\Gamma'$ for $Q(S_\varphi)$ as follows. For each vertex $v$ located at a grid point with coordinates $(i,j)$ in $\Gamma$, place the topmost, leftmost vertex of $q(v)$ at $h(i,j) = (2i+2j,-2i+2j)$. Now place $q(uv)$, for every edge $uv \in S_\varphi$, at the unit-area square that intersects both $q(u)$ and $q(v)$.

For the converse, suppose $\Gamma'$ is a unit-length embedding for $Q(S_\varphi)$. We will show this implies the existence of a unit-length embedding $\Gamma$ for $S_\varphi$. Without loss of generality, let the topmost vertex in the leftmost column of $\Gamma$ be located at the grid's origin. The function $h: \mathbb{Z}^2 \rightarrow \mathbb{Z}^2$ just defined is clearly bijective. Then, for each square $q(v)$ located at a unit-area square whose topmost, leftmost corner has coordinates $(i,j)$, $i,j$ even (for these are, by construction, the $C_4$ associated to vertices, not edges, of $S_\varphi$), place vertex $v$ at coordinates $h^{-1}(i,j)=(\frac{i-j}{4},\frac{i+j}{4})$ of an initially empty embedding $\Gamma$. Now link vertices $u,v$ by a unitary segment, in $\Gamma$, if there is a $C_4$ in $\Gamma'$ intersecting both $q(u)$ and $q(v)$, and $\Gamma$ is a unit-length embedding for $S_\varphi$.
\end{proof}

\begin{thm}
\textsc{PGR} is NP-complete for \{2,4\}-graphs.
\end{thm}
\begin{proof}
By Lemma~\ref{lemSquare}, \textsc{Bhatt-Cosmadakis} reduces to \textsc{PGR} for \{2,4\}-graphs, hence the latter is NP-complete.
\end{proof}

Again, since there are no trees without degree-1 vertices, the problem on \{2,4\}-graphs cannot be restricted to trees.

\subsection*{\textbf{\{1,3\}-graphs}}

The idea is basically the same. We introduce an appropriate gadget (one that is a \{1,3\}-tree, in this case) and an associated transformation that, given an extended skeleton $S_\varphi$, produces a \{1,3\}-tree $Q(S_\varphi)$ with the same immersibility.

\begin{figure}[!t]
\begin{center}

\setlength{\unitlength}{1mm}

\includegraphics[scale=0.69]{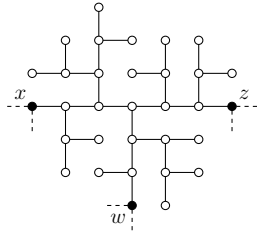}

\caption{\label{fig3plug}The \{1,3\} gadget (three-plug tree).}

\end{center}
\end{figure}

The gadget we employ is the one shown in Figure~\ref{fig3plug}. We call it the \emph{three-plug tree}. As usual, interconnectors are the labeled vertices in the figure.

We define the \emph{three-plug substitution} analogously to the double-ladder substitution, only replacing the double ladder with the three-plug tree. The three-plug substitution has an odd characteristic, though. Since the three-plug tree only presents 3 interconnectors, the input of a three-plug substitution is restricted to graphs with maximum degree not greater than 3. We want to show that \textsc{Bhatt-Cosmadakis} reduces polynomially to the problem of deciding whether a \{1,3\}-tree is a partial grid. But extended skeletons, which are the input of the \textsc{Bhatt-Cosmadakis} problem, present degree-4 vertices, hence the three-plug substitution cannot be applied to extended skeletons \emph{directly}.

This apparent hindrance is solved by first transforming the extended skeleton into a \{1,2,3\}-tree---with its same immersibility---via Gregori's U-tree substitution. Then, the resulting U-tree-transformed skeleton $U(S_\varphi)$, which has no 4-degree vertices, can be submitted to the three-plug substitution uneventfully, obtaining a \{1,3\}-tree $T(U(S_\varphi))$, still with the same original immersibility.

\begin{lem}\label{lem3plug}Three-plug substitution---with appropriately chosen interconnectors---preserves the immersibility of U-tree-transformed extended skeletons.
\end{lem}
\begin{proof}
Just like in the double ladder, interconnectors in the three-plug tree will always appear in the same circular permutation. This could have posed a problem to the desired immersibility preservation of the process, were it not for the fact that we know of a consistent orientation for U-tree-transformed skeletons (by Fact~\ref{pro:Utree_known_orientations}). Thus, with active interconnectors of the three-plug trees chosen appropriately, and because two adjacent three-plug trees will always occupy a rigid $7\times 14$ rectangle, the graph $T(U(S_\varphi))$ resulting from a three-plug substitution on $U(S_\varphi)$ will admit a unit-length embedding if and only if $U(S_\varphi)$ does.
\end{proof}

\begin{thm}
\textsc{PGR} is NP-complete for \{1,3\}- and \{1,3,4\}-trees.
\end{thm}
\begin{proof}
Same strategy here. By Lemma~\ref{lem3plug} and the fact that U-tree substitution preserves the immersibility of extended skeletons (proved by Gregori~\cite{IPL89}), \textsc{Bhatt-Cosmadakis} reduces to \textsc{PGR} for \{1,3\}-trees, hence the latter problem is NP-complete. The NP-completeness for \{1,3,4\}-trees follows, by the superset property.
\end{proof}

\begin{thm}
\textsc{PGR} is NP-complete for strictly binary trees.
\end{thm}
\begin{proof}
A strictly binary tree is a connected, acyclic graph whose vertices fall in one of three categories: (i)~1-degree vertices (the tree's \emph{leafs}); (ii)~a single 2-degree vertex (the tree's \emph{root}); and (iii)~3-degree vertices (the \emph{internal vertices}). After transforming an extended skeleton $S_\varphi$ into $T(U(S_\varphi))$ via three-plug substitution, the resulting graph comprises a series of interconnected three-plug-trees. Take any 1-degree vertex $v$ that sits next to a non-used grid point in the known embedding of the three-plug tree (say, for example, the topmost vertex in Figure~\ref{fig3plug}) and give it a new neighbor $w$ not yet in the graph. Vertex $v$ has become a 2-degree vertex, and the whole graph $T(U(S_\varphi))$ is now a strictly binary tree, rooted in $v$, with the same immersibility as $S_\varphi$. This completes the proof.
\end{proof}

\subsection{Polynomial cases}

\subsection*{\textbf{\{1,2\}-graphs}}
Trivial. A path on $n$ vertices can always be laid out on a straight line of a $1\times n$ grid, and any even cycle on $2k$ vertices can be embedded on a $2\times k$ grid. Odd cycles are not bipartite and therefore cannot be partial grids.

\subsection*{\textbf{\{3,4\}-graphs}}
\begin{thm}\label{thm34}
No unit-length embedding exists for \{3\}-, \{4\}- or \{3,4\}-graphs.
\end{thm}
\begin{proof}
Suppose there is a unit-length embedding $\Gamma$ for a graph with no vertices of degree 1 or 2. Let $v$ be the topmost vertex in the leftmost column of $\Gamma$. Since all other vertices are placed below or to the right of $v$, $v$ can have at most 2 neighbors, a contradiction.
\end{proof}

\subsection*{\textbf{\{1,4\}-graphs}}

\begin{figure}[!t]
\begin{center}

\setlength{\unitlength}{1mm}

\includegraphics[scale=0.69]{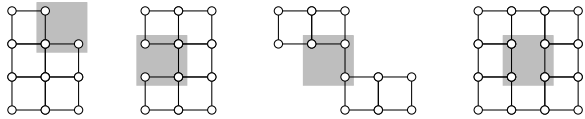}

\caption{\label{figD14proof}Proof of Theorem~\ref{thm14}: examples of incomplete unit-area squares $\sigma$ present in connected, non-grid, partial grids.}

\end{center}
\end{figure}

\begin{thm}\label{thm14}
A \{1,4\}-graph is a partial grid if and only if its degree-4 vertices induce a grid. \textsc{PGR} is therefore polynomial for \{1,4\}-graphs.
\end{thm}
\begin{proof}
Let $G$ be a connected \{1,4\}-graph. If the subgraph of $G$ induced by all its vertices of degree 4 is a grid, then there is always a unit-length embedding for $G$, in which the degree-4 vertices occupy all points of an $M\times N$ rectangle, surrounded by the $2(M+N)$ degree-1 vertices, which are necessarily adjacent to the vertices in the boundaries of such rectangle.

Now, let $\Gamma$ be a unit-length embedding for $G$, and let $G'$ be the graph induced by all degree-4 vertices of $G$. Since $G$ is a partial grid, $G'$ is a partial grid as well. Moreover, $G'$ must be connected, since $G$ is itself connected and the vertices in $G \setminus G'$ have degree 1. Suppose, by contradiction, that $G'$ is a connected partial grid that is not a grid graph (i.e.~the image of its grid mapping does not correspond to all the points and segments of an $M\times N$ rectangle in the grid). This hypothesis implies the existence of some unit-area square $\sigma$ (see Figure~\ref{figD14proof}), in $\Gamma$, containing at least 2 but no more than 3 edges of $G'$. Without loss of generality, let $u,v \in G'$ be incident to two such edges and placed at the extremes of a diagonal of $\sigma$. Since $u$ and $v$ have degree 4 in $G$, the two other diagonally opposed corners of $\sigma$ must correspond to vertices $s,t \in G$ which are necessarily adjacent to both $u$ and $v$. Thus, the degree of $s$ and $t$, in $G$, is at least 2, hence exactly 4, therefore $s$ and $t$ must belong to $G'$ as well. As a result, $\sigma$ contains 4 edges $us,sv,vt,tu$ of $G'$, a contradiction.

A polynomial-time recognition of grids can be achieved as follows. First, locate the 4 vertices of degree 2 and the $2(M+N)-4$ vertices of degree 3 present in the graph. They define the boundaries of an $M\times N$ rectangle in the grid. Now, recursively place each degree-4 vertex at the fourth point of a unit-area square already containing two of its neighbors (diagonally opposed in the grid) and one of its non-neighbors. Repeat this procedure inwardly, starting from the rectangle corners, until the vertices of degree 4 have matched the inner points of the rectangle (in which case the graph is a grid) or until such matching does not exist (in which case it is not).
\end{proof}

\section{Three-dimensional partial grids}\label{sec3d}

A \emph{3d grid} $G_{K \times L \times M}$ has as vertex set the points $\{1,\ldots,K\} \times \{1,\ldots,L\} \times \{1,\ldots,M\}$. Two vertices are adjacent if their distance is exactly $1$. A \emph{3d partial grid} is any subgraph (not necessarily induced) of a 3d grid. The 3d grids and 3d partial grids are bipartite, but not necessarily planar. The 3d \textsc{Partial Grid Recognition} problem (3d-\textsc{PGR}) consists of deciding whether a given graph is a 3d partial grid. Its NP-completeness was previously unknown.

Vertices of a 3d grid have degree at most $6$. Thus, a complete dichotomy would need to consider $2^6-1=63$ possible nonempty subsets. Although providing a complete dichotomy for three-dimensional partial grids is beyond the scope of this paper, it is perhaps surprising that the techniques developed for the two-dimensional case settle the complexity of all but $13$ out of those $63$ cases, as we show next.

We define the prism of a graph $G$ as the simple graph with vertices\linebreak $V(G) \times \{0,1\}$ and edge $(u,i)(v,j)$ if (i) $u=v$ or (ii) $i=j$ and $uv \in E(G)$. A partial grid and the corresponding prism are illustrated in Figure~\ref{fig:prism}. The following theorem uses the prism operator to interrelate the immersibilities of partial grids and 3d partial grids.

\begin{figure}[!t]
\begin{center}

\setlength{\unitlength}{1mm}

\includegraphics[scale=0.69]{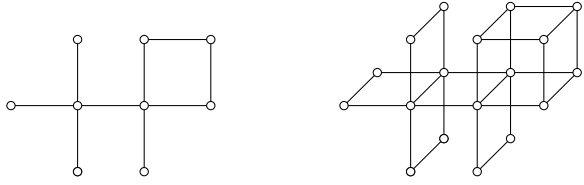}

\caption{\label{fig:prism}A partial grid and its corresponding prism.}

\end{center}
\end{figure}

\begin{thm} \label{thm:prism}
A graph $G$ is a partial grid if and only if the prism of $G$ is a 3d partial grid.
\end{thm}
\begin{proof}
If $G$ is a partial grid, then the prism of $G$ is a 3d partial grid because the three dimensional embedding can be obtained by placing two copies of the two-dimensional embedding on adjacent parallel planes. To prove the converse, we note that the $C_4$ graphs sharing opposing edges form a rigid structure that can only be embedded in a 3d grid space as two copies of $G$ on parallel planes.
\end{proof}

The previous theorem, along with the superset property, can leverage our previous (two-dimensional) results to show that $32$ out of the $63$ nonempty subsets of $\{1,2,3,4,5,6\}$ are NP-complete.

Given a set $D$ and a positive integer $k$, we define $D+k = \{d+k : d \in D\}$.

\begin{cor} \label{cor:npcprism}
If \textsc{PGR} is NP-complete for $D$-graphs then 3d-\textsc{PGR} is NP-complete for $(D+1)$-graphs.
\end{cor}
\begin{proof}
The problem is clearly in NP, since the embedding provides a polynomial certificate. To prove the NP-hardness we reduce \textsc{PGR} to 3d-\textsc{PGR} using the prism graph. Note that if $G$ is a $D$-graph, then the prism of $G$ is a $(D+1)$-graph. Correctness follows from Theorem~\ref{thm:prism}.
\end{proof}

Next, we present an extension of Corollary~\ref{cor:npcprism}, which proves the NP-{com\-plete\-ness} of $7$ additional subsets.

\begin{cor} \label{cor:npchairyprism}
If \textsc{PGR} is NP-complete for $(D_1 \cup D_2)$-graphs then 3d-\textsc{PGR} is NP-complete for $(\{1\} \cup (D_1+1) \cup (D_2+2))$-graphs.
\end{cor}
\begin{proof}
The proof follows from the prism construction with new vertices of degree $1$ appended to the vertices with degree in $D_2$.
\end{proof}

Graphs with degree at most $2$ are trivial. The following theorem is analogous to Theorem~\ref{thm34} and shows that the problem is polynomial for graphs where all vertices have degree $4$ and above.

\begin{thm} \label{thm:degreeatmost3}
A 3d partial grid has some vertex of degree at most $3$.
\end{thm}
\begin{proof}
Suppose there is a unit-length embedding $\Gamma$ for a graph with no vertices of degree 1, 2 or 3. Let $v$ be the topmost vertex in the leftmost column of the front most plane of $\Gamma$. Vertex $v$ can have at most 3 neighbors, a contradiction.
\end{proof}

The three-dimensional version of $\{1,6\}$-graphs can be decided polynomially.

\begin{thm}
A \{1,6\}-graph is a 3d partial grid if and only if its degree-6 vertices induce a 3d grid. Thus, 3d-\textsc{PGR} is polynomial for \{1,6\}-graphs.
\end{thm}
\begin{proof}
The proof is analogous to that for \{1,4\}-graphs in the two-dimensional case. Here, if we suppose that a \{1,6\}-graph $G$ is a 3d partial grid but its degree-6 vertices induce a graph which is certainly a partial 3d grid but not actually a 3d grid, then the mandatory existence of an incomplete unit-volume cube on its 3d embedding (one without all 12 edges, but with at least 3 vertices not on the same face) will lead to a similar contradiction.
\end{proof}

\section{Conclusion and open problems}\label{secConclusion}

Table~\ref{figComplexityDichotomy} gives the full dichotomy into polynomial and NP-complete for the recognition of (two-dimensional) partial grids. Previous results are duly referenced. Note that, for every degree set $D \supseteq \{1\}$, the complexity classes for $D$-graphs and $D$-trees match. It is also noteworthy that the results herein obtained are sufficient to show that the problem remains NP-complete even when a consistent orientation for the input graph is provided.

\begin{table}
\begin{center}

\begin{tabular}{ccc}

\begin{tabular}{ccc}

~~$D$~~&~~$D$-graphs~~&~~$D$-trees~~\\

\hline

\{1\} & P & P \\
\{2\} & P & --- \\
\{3\} & P & --- \\
\{4\} & P & --- \\
\{1,2\} & \textbf{P} & P \\
\{1,3\} & NPC & \textbf{NPC} \\
\{1,4\} & \textbf{P} & P \\
\{2,3\} & \textbf{NPC} & --- \\


\end{tabular}

&~&

\begin{tabular}{ccc}

~$D$~&~$D$-graphs~&~$D$-trees~\\

\hline

\{2,4\} & \textbf{NPC} & --- \\
\{3,4\} & \textbf{P} & --- \\
\{1,2,3\} & NPC~\cite{IPL89} & \textbf{NPC}~\cite{IPL89} \\
\{1,2,4\} & NPC~\cite{IPL87} & \textbf{NPC}~\cite{IPL87} \\
\{1,3,4\} & NPC & NPC \\
\{2,3,4\} & NPC & --- \\
\{1,2,3,4\} & NPC~\cite{IPL87} & NPC~\cite{IPL87} \\
~~~&~~~&~~~\\


\end{tabular} \\

\end{tabular}

\caption{\label{figComplexityDichotomy}Full complexity dichotomy for \textsc{PGR} (NPC: NP-complete; P: polynomial; ---: the corresponding input does not exist). Bold letters indicate the base cases, wherefrom the other cases derived (by the superset/subset property).}

\end{center}
\end{table}

A natural question concerns the existence of \emph{robust} gadgets. A robust gadget $R$ \emph{always} preserves the immersibility of the original graph $G$, when the vertices of $G$ are replaced by copies of $R$. The gadgets introduced herein, while sufficient for the intended proofs, do not guarantee that the immersibility of the original graph $G$ is preserved when a consistent orientation of $G$ is unknown. The graph shown in Figure~\ref{figWindmill}(a), called the \emph{windmill} graph, is one such robust gadget.\footnote{Indeed the windmill could perfectly have been used to prove the NP-completeness of \textsc{PGR} for \{1,3,4\}-trees, had that result not come as a byproduct of the \{1,3\}- case (by the superset property).} Since each of the windmill ``arms''---one of which is highlighted in Figure~\ref{figWindmill}(a)---are independently tied to the windmill ``axis''---its \emph{center}---by an edge, it is possible that they interchange their positions so to allow for any desired circular permutation of the gadget's interconnectors. Consequently, the windmill tree does not impose any fixed, predefined positioning of the neighborhood of each vertex being replaced, and the preservation of the original graph's immersibility is guaranteed. The proposed question asks whether or not there exist robust gadgets for degree sets other than the windmill's \{1,3,4\}.

\begin{figure}
\begin{center}

\setlength{\unitlength}{1mm}

\includegraphics[scale=0.69]{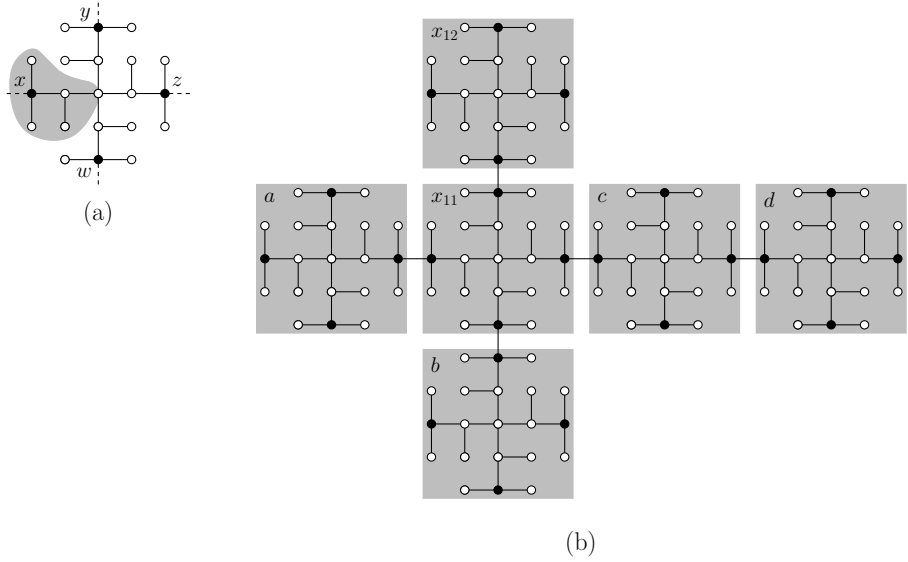}

\caption{\label{figWindmill}(a) The \{1,3,4\} gadget (windmill tree). (b) Windmill substitution.}

\end{center}
\end{figure}

Another question worth considering is how the complexities get affected by allowing edges with length up to $k > 1$.

Finally, completing the complexity dichotomy for the three-dimensional case (given in Table~\ref{fig3dComplexityDichotomy}) is a challenging problem, due to the rising number of applications employing three-dimensional layouts and to its intriguing theoretical appeal. In particular, so far we do not know of a complexity-separating degree set $D$ for which 3d-\textsc{PGR} is polynomial for $D$-trees but NP-complete for $D$-graphs.

\begin{table}
\begin{center}
\begin{tabular}{c|cccccccc}
 & $\emptyset$ & $\{4\}$ & $\{5\}$ & $\{6\}$ & $\{4,5\}$ & $\{4,6\}$ &
$\{5,6\}$ & $\{4,5,6\}$\\
\hline
$\emptyset$ & --- & P & P & P & P & P & P & \textbf{P}\\
$\{1\}$ & P & ? & ? & \textbf{P} & \textbf{NPC}$^2$ & \textbf{NPC}$^2$ & ? & NPC$^2$\\
$\{2\}$ & P & \textbf{NPC}$^1$ & ? & ? & NPC$^1$ & NPC$^1$ & ? & NPC$^1$\\
$\{3\}$ & ? & \textbf{NPC}$^1$ & \textbf{NPC}$^1$ & ? & NPC$^1$ & NPC$^1$ & NPC$^1$ & NPC$^1$\\
$\{1,2\}$ & \textbf{P} & NPC$^1$ & \textbf{NPC}$^2$ & ? & NPC$^1$ & NPC$^1$ & NPC$^2$ & NPC$^1$\\
$\{1,3\}$ & ? & NPC$^1$ & NPC$^1$ & \textbf{NPC}$^2$ & NPC$^1$ & NPC$^1$ & NPC$^1$ & NPC$^1$\\
$\{2,3\}$ & ? & NPC$^1$ & NPC$^1$ & ? & NPC$^1$ & NPC$^1$ & NPC$^1$ & NPC$^1$\\
$\{1,2,3\}$ & ? & NPC$^1$ & NPC$^1$ & NPC$^2$ & NPC$^1$ & NPC$^1$ & NPC$^1$ & NPC$^1$\\
\end{tabular}

\caption{\label{fig3dComplexityDichotomy}Known complexity dichotomy for the three-dimensional case (NPC$^1$: NP-complete due to Corollary~\ref{cor:npcprism}; NPC$^2$: NP-complete due to Corollary~\ref{cor:npchairyprism}; P: polynomial; ?: open case; ---: the corresponding input does not exist). Each cell states the complexity of 3d-\textsc{PGR} restricted to $D$-graphs, where $D$ is the union of the sets associated to the column and the row that contain the cell. Bold letters indicate the base cases, wherefrom the other cases derive (by the superset/subset property).}

\end{center}
\end{table}

\section*{Acknowledgements}
The authors would like to thank professors Lucia Draque Penso and Dieter Rautenbach for the insightful discussions.

\bibliographystyle{amsplain}

\end{document}